%% file: main.tex
\documentclass[runningheads]{llncs}
\input{defs.tex}

\commentfalse 

\versionShorttrue 
\versionShortfalse 

\highlightShortLongDiffstrue 
\highlightShortLongDiffsfalse

\begin{document}
\setlength{\abovedisplayskip}{1pt} 
\setlength{\belowdisplayskip}{3pt}
\title{Quadratic-Curve-Lifted Reed-Solomon Codes} 
%
%
\author{Hedongliang Liu\inst{1}\orcidID{0000-0001-7512-0654} \and
Lukas Holzbaur\inst{1}\orcidID{0000-0002-8048-3051} \and
Nikita Polyanskii\inst{2}\orcidID{0000-0003-2626-8139} \and
Sven Puchinger\inst{3}\orcidID{0000-0002-7474-2678} \and
Antonia Wachter-Zeh\inst{1}\orcidID{0000-0002-5174-1947}
\thanks{The work of H.~Liu has been supported by a German Israeli Project Cooperation (DIP) grant under grant no.~PE2398/1-1 and KR3517/9-1. The work of L.~Holzbaur, N.~Polyanskii, and A.~Wachter-Zeh was supported by the German Research Foundation (Deutsche Forschungsgemeinschaft, DFG) under Grant No. WA3907/1-1.
}}
\authorrunning{H.~Liu et al.}
%
\institute{Technical University of Munich, Germany
  \and IOTA Foundation, Germany
  \and HENSOLDT Sensors GmbH, Germany
}
\maketitle              
\begin{abstract}

Lifted codes are a class of evaluation codes attracting more attention due to good locality and intermediate availability. In this work we introduce and study \emph{quadratic-curve-lifted Reed-Solomon (QC-LRS) codes}, which is a class of bivariate evaluation codes and the codeword symbols whose coordinates are on a quadratic curve form a codeword of a Reed-Solomon code. We first develop a necessary and sufficient condition on the monomials which form a basis of the code.
Based on the condition, we give upper and lower bounds on the dimension and show that the asymptotic rate of a QC-LRS code over $\mathbb{F}_q$ with local redundancy $r$ is $1-\Theta(q/r)^{-0.2284}$. Moreover, we provide analytical results on the minimum distance of this class of codes and compare QC-LRS codes with lifted Reed-Solomon codes by simulations in terms of the local recovery capability against erasures. For short lengths, QC-LRS codes have better performance in local recovery for erasures than LRS codes of the same dimension.
\keywords{Lifted Codes \and Reed-Solomon Codes \and Quadratic Curves \and Locality \and Dimension.}
\end{abstract}
\section{Introduction}
\label{sec:intro}

Lifted codes were introduced by Guo, Kopparty and Sudan~\cite{guo2013new} as evaluation codes obtained from multivariate polynomials over finite fields.
Informally, the key property of these codes is that the restriction to any affine subspace of fixed dimension
of the evaluation space is a codeword of a fixed base code.
A setting of particular interest, referred to as lifted Reed-Solomon (LRS) codes, is given by lifted codes where each $1$-dimensional affine subspace is a codeword of an RS code.
This can be viewed as a generalization of the well-known Reed-Muller codes.
A surprising advantage of LRS codes is that they achieve much larger asymptotic code rate as the field size grows compared to Reed-Muller (RM) code.
The dimension of LRS codes is analyzed via the number of \emph{good monomials}, i.e., the number of multi-variate monomials that result in a codeword of the base code when evaluated on any fixed line. The linear span of these good monomials is shown to generate the entire lifted code. The study of these codes was continued in \cite{polyanskii2019trivariate,holzbaur2020lifted}, where tight asymptotic bounds on the rate were derived.

The seminal paper \cite{guo2013new} gave rise to a number of related concepts and code constructions. The works
\cite{Li2019lifted-multiplicity,wu2015revisiting,holzbaur2021lifted}
consider lifting of \emph{multiplicity codes} \cite{kopparty2014high}, another class of codes with good locality properties. \emph{Degree-lifted} codes were introduced in \cite{eli2013degree-lifted}, where each codeword polynomial is constructed as the product of the uni-variate polynomials in the base code with an additional constraint on the total degree.
A class of lifted codes based on code automorphisms was introduced in \cite{guo2016high-rate}.
Codes constructed from all bivariate polynomials, evaluated on the Hermitian curve, such that the restriction to any line agrees with some low-degree univariate polynomial on the points of the Hermitian curve intersected with that line were analyzed in \cite{lopez2021hermitian} and named \emph{Hermitian-lifted codes}.
A variant of lifted codes that utilizes the trace operation to obtain binary codes with good locality properties was introduced in \cite{hastings2020wedge} and coined \emph{wedge-lifted codes.}
\new{Thanks to the comments from an anonymous reviewer, we noticed that the recent work~\cite{lavauzelle2020weighted} gave a more general definition of lifted codes with curves, which is called \emph{weighted lifted codes}. The QC-LRS code (\cref{def:curve-liftedRScodes}) studied in this work is coincidentally identical to \cite[Def.~IV.1]{lavauzelle2020weighted} with $\eta=2$.}

\subsection{Main Contribution and Organization}
\label{sec:motivation}
All works mentioned above consider the restriction to linear subspaces or code automorphisms.
Our work provides a class of evaluation codes whose local recovery sets correspond to \new{a} set of quadratic curves. \new{The advantage of this construction is that there is a much larger number of recovery sets for each codeword symbol, however, these recovery sets do no longer (necessarily) intersect in only one position.}
This work is devoted to the analysis of the rate and distance of these codes, as well as their local recovery capability compared to LRS/RM codes.

We first investigate the dimension of QC-LRS codes. Since QC-LRS codes are evaluation codes, we analyze the dimension by first deriving the necessary and sufficient condition on the good monomials, where we take similar approaches as in~\cite{hastings2020wedge}, and then by showing that these good monomials form a basis of the code as in \cite{Li2019lifted-multiplicity}.
By quantifying the bad monomials following the approach for LRS codes from~\cite{holzbaur2020lifted}, we derive upper and lower bounds on the dimension of QC-LRS codes over $\Fq$ with $q$ being a power of two. The asymptotic rate of QC-LRS codes over $\Fq$ with local redundancy $r$ is shown to be $1-\Theta((q/r)^{-0.2284})$.
\new{The approach in this paper gives a more precise estimation of the dimension than that in \cite{lavauzelle2020weighted}, which studied a more general definition of lifted codes with curves of arbitrary degree.}
To study the advantage of more local groups given by the new notion of QC-LRS codes than LRS codes, we compare between LRS codes and QC-LRS codes the failure probability of locally
recovering codeword symbols from erasures. The simulation results show that for the blocklength $64$  and under the same code dimension, QC-LRS codes have similar or better performance than LRS codes.

The organization of this paper is as follows: \cref{sec:preliminaries} introduces the notations used throughout the paper and some basics,
which are required in the proofs of the main results.
In \cref{sec:CLC} we formally define the QC-LRS codes and present results on the dimension and distance.
\versionShortLong{Due to the page limit, we omit some proofs and we refer to the full version of this paper~\cite{liu2021quadratic} for an extensive justification.}{}
Section 4 presents the comparison on the failure probability of local
recovery from erasures by QC-LRS and LRS codes.

\section{Preliminaries}
\label{sec:preliminaries}
\label{sec:notations}
Denote the set of integers $\{a,\dots,b\}$ by $[a,b]$ and by $[b]$ if $a=1$. A finite field of size $q$ is denoted by $\Fq$. The integer ring of size $q$ is denoted by $\mathbb{Z}_q$.
Let $\deg:\Fq[x]\to \new{\bbN}$ be the degree function of univariate polynomials. For any $f=\sum_{i=0}^{q-1}f_ix^i$, $\deg(f)=\max\{i|f_i\neq 0\}$.
For non-negative integers $a,b\in\bbN$ with \emph{binary representations} $a=(a_0,\dots, a_{\ell-1})_2$, $b=(b_0,\dots, b_{\ell-1})_2$, we say that \emph{$a$ lies in the 2-shadow of $b$}, denoted by $a\leq_2 b$, if $a_i\leq b_i,\ \forall i \in[0,\ell-1]$. The bit $a_{\ell-1}$ is the most significant bit in the binary representation of $a$.
For a bi-variate function $f:\Fq^2\to\Fq$ and a set $D\subset\Fq^2$, let $f|_{D}$ denote the restriction of $f$ to the domain $D$. If $D$ is the set of points corresponding to the roots in $\Fq^2$ of a bi-variate function $\phi:\Fq^2\to\Fq$, i.e., $D=\{(x,y)\in\Fq^2\ |\ \phi(x,y)=0\}$, we denote by $f|_{\phi}$ the restriction of $f$ to the \emph{curve} $\phi$.
A bivariate function $\phi:\Fq^2\to\Fq$ is a \emph{quadratic function} or \emph{quadratic curve} if it is in the form
  $\phi(x,y)=y+\alpha x^2+\beta x+\gamma \ ,$
where $\alpha,\beta,\gamma \in \Fq$.

Define an operation $(\Mod^*\ q)$ that takes a non-negative integer and maps it to an element in $[0, q-1]$ as follows
\begin{align*}
  a\ (\Mod^*\ q) :=
  \begin{cases}
    a, & \text{if } a\leq q-1\\
    q-1, & \text{if } a\ (\Mod\ q-1) =0, a\neq 0\\
    a \ (\Mod\ q-1), & \text{else}
  \end{cases}
\end{align*}
It can be readily seen that if $a\modstarq=b$, then $x^a=x^b\ (\Mod\ x^q-x)$ in $\Fq[x]$.


\begin{lemma}[Lucas' Theorem~\cite{LucasTheorem}]
  \label{lem:lucasThm}
  Let $p$ be a prime and $a,b\in\bbN$ be written in $p$-ary representations $a=(a_0,\dots, a_{\ell-1})_p$, $b=(b_0,\dots, b_{\ell-1})_p$. Then
  \begin{align*}
    \binom{a}{b}=\prod_{i=1}^\ell\binom{a_i}{b_i}\mod p\ .
  \end{align*}
  If $p=2$, then $\binom{a}{b}=1$ if and only if $b\leq_2 a$.
\end{lemma}

\begin{lemma}[Combinatorial Nullstellensatz~{\cite[Theorem 1.2]{alon1999combinatorial}}]
  \label{lem:nullstellansatz}
  Let $\F$ be an arbitrary field, and let $f(x_1, \dots,x_m)$ be a multivariate polynomial in $\F[x_1,\dots,x_m]$ \new{of degree} $\deg(f)=\sum_{i=1}^m t_i$, where each $t_i$ is a non-negative integer, and suppose the coefficient of $\prod_{i=1}^m x_i^{t_i}$ in $f$ is nonzero. Then, if $S_1,\dots, S_m$ are subsets of $\F$ with $|S_i|> t_i$, there are $s_1\in S_1,\dots, s_m\in S_m$ so that
  \begin{align*}
    f(s_1,\dots,s_m)\neq 0.
  \end{align*}
\end{lemma}

\section{Quadratic-Curve-Lifted Reed-Solomon Codes}
\label{sec:CLC}
In this section, we first give a general definition of curve-lifted Reed-Solomon codes and present our results on a specific class of codes, the QC-LRS codes, with restriction to quadratic curves.
\begin{definition}[Curve-Lifted Reed-Solomon Codes]
  \label{def:curve-liftedRScodes}
   Let $q$ be a power of $2$ and $\Phi$ be a set of bi-variate functions. A curve-lifted Reed-Solomon code is defined by
  \begin{align*}
    \cC_{q}(\Phi,d):=\{f:\Fq^2\to\Fq\ |\ \deg (f|_{\phi})<d,\forall \phi\in\Phi \}\ .
  \end{align*}
\end{definition}
In order to investigate the dimension of curve-lifted RS codes, we introduce the \emph{good monomials}\footnote{\new{This is a short notation to easily address these monomials later. There is no bias on the performance of the monomials.}}, which is a tool also used in studying LRS codes in~\cite{guo2013new,holzbaur2020lifted}.
\begin{definition}[{$(\Phi,d)^*$}-good monomial]
  \label{def:goodMonomials}
  Given a set $\Phi$ of bi-variate functions, a monomial $m(x,y)=x^ay^b$ is \emph{$(\Phi,d)^*$-good} if $\deg (m|_{\phi})<d, \forall \phi \in\Phi$. The monomial is \emph{$(\Phi,d)^*$-bad} otherwise.
\end{definition}

In the following let $\Phi$ be the set of all quadratic functions\footnote{\new{This set is a subset of quadratic curves, which are often referred as \emph{conics} in geometry. This set is also identical to the set of \emph{affine $\eta$-lines} defined in \cite{lavauzelle2020weighted} with $\eta=2$.}} over $\Fq$, i.e.,
\begin{align}\label{eq:Phi}
  \Phi:=\{\phi(x,y)=y+\alpha x^2+\beta x+\gamma, \forall \alpha,\beta,\gamma \in \Fq\}\
  \end{align}
and we present the results on QC-LRS codes.

The following~\cref{lem:goodCond} gives a necessary and sufficient condition such that a monomial $m(x,y)=x^ay^b$ is $(\Phi,d)^*$-good.
\begin{lemma}
  \label{lem:goodCond}
  Let $q$ be a power of $2$, $\Phi$ be the set of all quadratic functions over $\Fq$ and $a,b<q$ be integers. A monomial $m(x,y)=x^ay^b$ is $(\Phi,d)^*$-good if and only if
  \begin{align}\label{eq:good_mono_suf_nec_cond}
    2i+j+a \ (\Mod^*\ q) < d,\ \forall i\leq_2 b,\ j\leq_2 b-i\ .
  \end{align}
\end{lemma}
\versionShortLong{}{
\begin{proof}
  We write a quadratic function $\phi\in\Phi$ as $\phi(x,y)=y+\alpha x^2+\beta x+\gamma $. Then the monomial restricted to curve $\phi$ can be written as
  \begin{align*}
    m|_{\phi}(x)&=x^a(\alpha x^2+\beta x+\gamma)^b\\
             &=x^a\sum_{i=0}^b\binom{b}{i}\alpha^ix^{2i}\cdot (\beta x+\gamma)^{b-i}\\
             &=\sum_{i=0}^b\binom{b}{i}\alpha^ix^{2i+a}\cdot \sum_{j=0}^{b-i}\binom{b-i}{j} \beta^jx^j\cdot \gamma^{b-i-j}\\
                &\overset{(*)}{=} \sum_{i\leq_2 b}\alpha^ix^{2i+a}\cdot \sum_{j\leq_2 b-i} \beta^jx^j\cdot \gamma^{b-i-j}\\
    &=\sum_{i\leq_2 b}\ \sum_{j\leq_2 b-i}\alpha^i\cdot \beta^j\cdot \gamma^{b-i-j}\cdot x^{2i+j+a}
  \end{align*}
  where the equality $(*)$ follows from the Lucas' Theorem (\cref{lem:lucasThm}). If the condition~\eqref{eq:good_mono_suf_nec_cond} in the statement holds, then $\deg m_\phi(x)<d$, and the ``if'' direction is proved.

  Denote $m|_\phi^*(x)= m|_{\phi}(x)\ \Mod\ x^q-x$. The coefficient of $x^s$ in $m|_\phi^*(x)$ is
    \begin{align*}
      [x^s]m|_\phi^* = \sum_{\substack{i\leq_2 b,\ j\leq_2 b-i\\2i+j+a\modstarq =s}}\alpha^i\cdot \beta^j\cdot \gamma^{b-i-j}
    \end{align*}
    We can see this as a polynomial in $\alpha,\beta,\gamma$. Assume the condition~\eqref{eq:good_mono_suf_nec_cond} does not hold but $m(x,y)$ is $(\Phi,d)^*$-good, i.e., for any $s\geq d$, $[x^s]m|_\phi^*$ is not a zero polynomial but equal to $0$ evaluated at all $(\alpha,\beta,\gamma)\in\Fq^3$. However, by~\cref{lem:nullstellansatz}, since $i,j,b-i-j<q$, there exists some $(\alpha_0,\beta_0,\gamma_0)$ such that $[x^s]m|_\phi^*\neq 0$. By contradiction it can be seen that the condition~\eqref{eq:good_mono_suf_nec_cond} is also a necessary condition.
\qed \end{proof}}
\subsection{Dimension of Quadratic-Curve-Lifted RS Codes}
The first important result is that the dimension of the code is exactly the number of good monomials, which we present in~\cref{thm:dim_good}.
In order to show that, we first discuss in the following lemma a special case that will be excluded in the proof of~\cref{thm:dim_good}. Due to space limitations, we leave out the proof of the lemma here and refer to the full version of this paper~\cite{liu2021quadratic}.
\begin{lemma}\label{lem:difference_q-1}
  Consider two monomials $m_1(x,y)=x^{q-1}y^b$ and $m_2(x,y)=y^b$ with $0\leq b\leq q-1$ and a polynomial $P(x,y)$ containing $m_1$ and $m_2$, i.e.,
  \begin{align*}
    P(x,y) &= (\xi_1 x^{q-1}y^{b}+ \xi_2 y^b) + P'(x,y)
  \end{align*}
  where $\xi_1,\xi_2\neq 0$ and $P'(x,y)$ does not contain $m_1$ or $m_2$. Then, $P$ is $(\Phi,d)^*$-bad for any $d\leq q-1$.
\end{lemma}
\versionShortLong{}{
\begin{proof}
  Consider $P$ restricted to the curve $\phi:y=\gamma$ for some $\gamma\in \Fq$, i.e., $\alpha=\beta=0$,
\begin{align*}
    P|_\phi(x) &= \gamma^b (\xi_1 x^{q-1}+ \xi_2) + P'(x,y=\gamma) \ .
\end{align*}
First, observe that for this choice of $\alpha,\beta$ this polynomial is of $x$-degree at most $q-1$ and we are only interested in the coefficient of $x^{q-1}$. Further, the only monomials of $P'(x,y)$ that contribute to this coefficient are of the form $\xi' x^{q-1} y^{b'}$ with $\xi'\neq 0$, since we replace $y$ with a function that is of $x$-degree zero, i.e., with $\gamma$. Since $P'(x,y)$ does not contain the monomials of the pair $m_c, m_d$ by definition, we conclude that $b' \neq b$. Now consider the coefficient of $x^{q-1}$ in $P|_\phi(x)$
\begin{align*}
    [x^{q-1}]P|_\phi = \underbrace{\gamma^b \xi_1}_{\text{from }m_1+m_2} + \underbrace{\gamma^{b'} \xi' + \ldots}_{\text{from } P'(x,y)} \ .
\end{align*}
We view this as a polynomial in $\gamma$. Since $b\neq b'$ and $\xi_1,\xi'\neq 0$ this is not the all-zero polynomial. Also, as $b,b'\in[q-1]$ this is a polynomial of degree $\leq q-1$ in $\gamma$. By~\cref{lem:nullstellansatz}, there exists $\gamma\in\Fq$ such that $[x^{q-1}]P|_{\phi}\neq 0$, which means $P|_{\phi}$ is of degree $q-1$ for some $\gamma$. Therefore, $P$ is $(\Phi,d)^*$-bad according to~\cref{def:goodMonomials} for any $d\leq q-1$.
\qed \end{proof}}
\begin{theorem}[Dimension is the number of good monomials]
  \label{thm:dim_good}
  Let $d\leq q-1$ and $\Phi$ be the set of all quadratic functions. The QC-LRS code $\cC_{q}(\Phi,d)$ has dimension equal to the number of $(\Phi,d)^*$-good monomials over $\Fq$.
\end{theorem}

\begin{proof}
Assume a polynomial $P$ containing $(\Phi,d)^*$-bad monomials is $(\Phi,d)^*$-good.
Let $\cG$ and $\cB$ be subsets of indices of all $(\Phi,d)^*$-good and -bad monomials, respectively (assuming the monomials are ordered according to some order).
We can write $P$ as
\begin{align*}
    P=\sum_{c\in\cG} \xi_c x^{a_c}y^{b_c}+\sum_{c\in\cB} \xi_c x^{a_c}y^{b_c},
\end{align*}
with $\xi_c\in \F_q\setminus\{0\}$. Restricting $P$ to the quadratic curve $\phi:y=\alpha x^2+\beta x +\gamma$ is the univariate polynomial
\begin{align*}
    P|_{\phi}&=\sum_{c\in\cG\cup\cB} \xi_c x^{a_c}(\alpha x^2+\beta x+\gamma)^{b_c}\\
    &=\sum_{c\in\cG\cup\cB} \xi_c \sum_{i=0}^{b_c}\sum_{j=0}^{b_c-i}\binom{b_c}{i}\binom{b_c-i}{j}\alpha^i\cdot \beta^j\cdot \gamma^{b_c-i-j}\cdot x^{2i+j+a_c}.
\end{align*}

Let $P|_{\phi}^*=P|_{\phi}\ \Mod\ (x^{q}-x)$. Denote by $[x^s]P|_{\phi}^*$ the coefficient of $x^s$ in $P|_{\phi}^*$. \new{By Lucas' Theorem}
(see~\cref{lem:lucasThm}), we have
\begin{align}
    [x^{s}]P|_{\phi}^*=&\sum_{c\in\cG\cup\cB}\sum_{\substack{i\leq_2 b_c,\ j\leq_2 b_c-i\\2i+j+a_c\modstarq = s}}\xi_c\cdot  \alpha^i\cdot \beta^j\cdot \gamma^{b_c-i-j}\nonumber.
\end{align}
For $s\geq d$, the $(\Phi,d)^*$-good monomials do not contribute to these coefficients (see~\cref{def:goodMonomials}), therefore,
\begin{align}
    [x^{s}]P|_{\phi}^*{=}&\sum_{c\in\cB} \sum_{\substack{i\leq_2 b_c,j\leq_2 b_c-i\\2i+j+a_c\modstarq = s}}\xi_c\cdot \alpha^i\cdot \beta^j\cdot \gamma^{b_c-i-j}\quad \text{for }s\geq d .\label{eq:poly_alpha_beta_gamma}
\end{align}

We view $[x^{s}]P|_{\phi}^*$ as a trivariate polynomial in $\alpha,\beta,\gamma$.
Note that $P$ is $(\Phi,d)^*$-good only if
\begin{align}
\label{eq:P-good-cond}
    [x^{s}]P|_{\phi}^*\ (\alpha,\beta,\gamma)=0\ ,\quad \forall \alpha,\beta,\gamma\in\Fq, \forall s\geq d\ .
\end{align}

Now consider two bad monomials $x^{a_c}y^{b_c}$ and $x^{a_d}y^{b_d}$ with $c,d\in\cB$. 
Then the corresponding terms in~\eqref{eq:poly_alpha_beta_gamma} contributed by them can be \new{added up}
only if $\alpha^{i_c}\beta^{j_c}\gamma^{b_c-i_c-j_c}=\alpha^{i_d}\beta^{j_d}\gamma^{b_d-i_d-j_d}$, which is true if and only if
\begin{align*}
    \iff &\begin{cases}
      &i_c=i_d\\
      &j_c=j_d\\
      &b_c-i_c-j_c=b_d-i_d-j_d\\
      &2i_c+j_c+a_c \modstarq = 2i_d+j_d+a_d\modstarq
    \end{cases}\\
     \Longrightarrow &\begin{cases}
      &b_c=b_d\\
      &|a_c-a_d|=0 \textrm{ or } q-1\ .
    \end{cases}
\end{align*}
For the case $|a_c-a_d|=q-1$, such polynomials are bad according to~\cref{lem:difference_q-1}.
For the case $|a_c-a_d|=0$, we can conclude that the monomials $\alpha^i\beta^j\gamma^{b_c-i-j}$ are distinct for different pairs of $(a_c,b_c)$. Namely, \eqref{eq:poly_alpha_beta_gamma} is in its simplest form\footnote{No similar terms can be further combined.}.

Assume $\cB$ is non-empty. Since $\xi_c\neq 0$ for all $c$, \eqref{eq:poly_alpha_beta_gamma} is a non-zero polynomial.
By~\cref{lem:nullstellansatz}, since the variables $\alpha,\beta,\gamma\in\Fq$ and all exponents $i,j,b_c-i-j<q$, there exists some $\alpha_0,\beta_0,\gamma_0\in\Fq$, such that $[x^{s}]P|_{\phi}^*\neq 0$. This contradicts the assumption that $P$ is $(\Phi,d)^*$-good. This implies that \eqref{eq:P-good-cond} can be fulfilled only if $[x^{s}]P|_{\phi}^*$ is a zero polynomial, i.e., $\cB$ is empty. Hence, a polynomial $P$ is $(\Phi,d)^*$-good only if it only consists of good monomials.
\qed \end{proof}

\subsection{Estimation of the Dimension}
In this section we provide an analysis of the dimension of QC-LRS codes $\cC_{q}(\Phi,d=q-r)$, where $q=2^\ell$ and $r\in[q-1]$.
Recall from Lemma~\ref{lem:goodCond} that a monomial $m(x,y)=x^ay^b$ is $(\Phi,q-r)^*$-bad if and only if there exist $i\leq_2 b$ and $j\leq_2 b-i$ such that $2i+j+a \ (\Mod^*\ q)\geq q-r$.
We will first consider a slightly different definition of a bad monomial to simplify our arguments. Then, we derive upper and lower bounds on the number of  $(\Phi,q-r)^*$-bad monomials and further establish the results on the rate of QC-LRS codes.
\subsubsection{Counting $(\Phi,q-r)$-bad monomials: }
Let $q=2^\ell$ and $r\in[q-1]$. We say that a monomial $m(x,y)=x^ay^b$ (or the pair $(a,b)$) is $(\Phi,q-r)$-bad if and only if there exist $i\leq_2 b$ and $j\leq_2 b-i$ such that $2i+j+a \pmod{q} \geq q-r$. For an integer $t\geq 0$, we define
\begin{align}
    S_t(\ell)&=\bigg\{ (a,b)\in \mathbb{Z}_q^2\ :\ \begin{split}
        &\exists\ i\leq_2 b, j\leq_2 b-i , \\
        &\textrm{s.t.}\ 2i+j+a=q-r'+tq, \textrm{ for some }r'\in[r]
    \end{split}\ \bigg\}
    \label{eq:def_S_t}
\end{align}
For $1\leq r<q$ and $t\geq 3$, the set $S_t(\ell)$ is empty as $2i+j+a\leq i+b+a\leq 2b+a\leq 3(q-1)< q-r+tq$. Hence, if $x^ay^b$ is $(\Phi,q-r)$-bad, then $(a,b)\in S_0(\ell)\cup S_1(\ell)\cup S_2(\ell)$.

In what follows, we assume that  $1\leq r< q$ and attempt to derive some recursive relations on $S_0(\ell)$, $S_1(\ell)$ and $S_2(\ell)$. We have two observations in \cref{lem:go-down} and \cref{lem:down-by-one}.
\begin{lemma}\label{lem:go-down}
Let $q=2^\ell$ \new{and $r<\frac{q}{2}$}, $a=(a_0,\ldots,a_{\ell-1})_2$ and $b=(b_0,\ldots, b_{\ell-1})_2$. Define $a':=(a_0,\ldots,a_{\ell-2})_2$ and $b':=(b_0,\ldots, b_{\ell-2})_2$. If $(a,b)\in S_0(\ell)\cup S_{1}(\ell)\cup S_{2}(\ell)$, then $(a',b')\in S_0(\ell-1)\cup S_{1}(\ell-1)\cup S_2(\ell-1)$.
\end{lemma}
\versionShortLong{}{
\begin{proof}
The condition $(a,b)\in S_0(\ell)\cup S_1(\ell)\cup S_{2}(\ell)$ implies that there exist $i=(i_0,\ldots,i_{\ell-1})_2\le_2 b$ and $j=(j_0,\ldots,j_{\ell-1})_2\le_2 b-i$ such that $2i+j+a= q-r' \pmod{q}$, where $r'\in[r]$.
Let $i':=(i_0,\ldots, i_{\ell-2})_2$ and $j':=(j_0,\ldots,j_{\ell-2})_2$. Clearly $i'\le_2 b'$ and $j'\le_2 b'-i'$ and  $2i'+j'+a' = \frac{q}{2}-r' \pmod{\frac{q}{2}}$.
\qed \end{proof}}

\begin{lemma}\label{lem:down-by-one}
  For $t=1,2$, if $(a,b)\in S_t(\ell)$, then $(a,b)\in S_{t-1}(\ell)$. 
\end{lemma}
\versionShortLong{}{
\begin{proof}
  We first prove for $t=1$. The condition $(a,b)\in S_1(\ell)$ implies that there exists an $i\leq_2 b$ and an $j\leq_2 b-i$ with $2i+j+a = 2q-r'$ with $r'\in[r]$. The statement $(a,b)\in S_0(\ell)$ means that there exists $i'\leq_2 i$ and $j'\leq_2 j$ such that $2i'+j'+a \in[ q-r,q-1]$. Note that for $q-r\leq a\leq q-1$ the statement holds with $i'=j'=0$. Assuming $a< q-r$, we claim the existence of a pair $(i'\leq_2 i, j'\leq_2 j)$ such that $2i'+j'+a = q-r'$ which would imply the required statement. Such $i',j'$ can be found by the procedure in \cref{algo:ij-reduction} that replaces some ones in the binary representations of $i$ and $j$ by zeros so that $(2i+j)-(2i'+j')=q$.

  The procedure outputs the correct $i',j'$ for $2i+j>q$ if we enter the $\delta\leq 0$ else-part (\cref{line:deltaleq0}) in Step 2 at some point. Assume the contrary that this does not happen, meaning that the procedure output the all-zero $i',j'$ at the end. However, this implies that $\delta=q-(2i+j)>0$ which contradicts the condition that $2i+j+a\geq 2q-r$ while $a<q-r$.

  For $t=2$, given $i,j$ such that $2i+j+a=3q-r'$, which implies that $2i+j>q$, we can find $i',j'$ by \cref{algo:ij-reduction} such that $2i'+j'+a=2q-r'$. This complete the proof.
\qed
\end{proof}

\begin{algorithm}
\caption{Deduct $q$}\label{algo:ij-reduction}
\DontPrintSemicolon
\SetKw{Init}{Init:}
\KwIn{$i,j$}
\KwOut{$i',j'$}
\Init{$i'\leftarrow i,j'\leftarrow j, h\leftarrow \ell,\Delta\leftarrow 1$}\\
\uIf{$h=0$\label{line:h_equal_0}}{ \Return{$i',j'$}}
Let $h\leftarrow h-1$ and $\Delta\leftarrow 2\Delta$\label{line:Delta}\\
Compute $\delta\leftarrow \Delta-i'_{h-1}-j'_h$\\
\uIf{$\delta>0$}{$i'_{h-1}\leftarrow 0,\ j'_h\leftarrow 0$\\
\textbf{Go back to Line}~\ref{line:h_equal_0}}
\uElse{\label{line:deltaleq0}
Let $
                \begin{cases}
                    i'_{h-1}\leftarrow 0,\quad & \textbf{if }\Delta-i'_{h-1}=0\\
                    j'_{h}\leftarrow 0,\quad & \textbf{if }\Delta-j'_{h}=0\\
                    i'_{h-1}\leftarrow 0,j'_h\leftarrow 0, \quad & \textbf{if }\Delta-i'_{h-1}-j_h=0
                \end{cases}
$\\
\Return $i',j'$}
\end{algorithm}
\begin{example}
Consider the parameters $q=2^\ell=2^4, r'=2$. In the following we may also use the binary representation for the integers, e.g., $a=(a_0,\dots,a_{\ell-1})_2$. For the element $(a,b)=(12,14)=((0011)_2,(0111)_2)\in S_1(4)$ and $i=(0010)_2,j=(0101)_2$ such that $i\leq_2 b,j\leq_2 b-i$ and $2i+j+a=2q-r'=(01111)_2$, we can find the corresponding $i'\leq_2 i,j'\leq_2 j$ such that $2i'+j'+a=q-r'$ by \cref{algo:ij-reduction}:
\begin{enumerate}
    \item \textbf{Init:} $i'\leftarrow (0010)_2,j'\leftarrow (0101)_2$, $\Delta=1$ and $h\leftarrow 4$.
    \item Let $h\leftarrow 3, \Delta\leftarrow 2$, compute $\delta\leftarrow \Delta-i'_3-j'_4=2-1-1=0$.
    Since $\delta\not>0$ and $\Delta-i'_3-j'_4=0$, $i'_3\leftarrow 0$, $j'_4\leftarrow 0$ and output $i'=(0000)_2,j'=(0100)_2$.
\end{enumerate}
As $i'\leq_2i\leq b$, $j'\leq_2 j\leq_2 b-i$ and $2i'+j'+a=(0111)_2=q-r$, $(a,b)$ is in $S_0(\ell)$.
\end{example}}
It follows from Lemma~\ref{lem:down-by-one} that $x^ay^b$ is $(\Phi,q-r)$-bad if and only if $(a,b)\in S_0(\ell)$.

Based on the observations in \cref{lem:down-by-one} and \cref{lem:go-down}, we provide a recursive formula for computing the size of $S_t(\ell)$ for $t=0,1,2$.
\begin{lemma}\label{lem: key ingredient}
For $1\leq r< \frac{q}{2}$, it holds that
\begin{align*}
    |S_0(\ell)|&=3|S_0(\ell-1)|+|S_1(\ell-1)|,\\
    |S_1(\ell)|&=|S_0(\ell-1)|+|S_1(\ell-1)| + |S_2(\ell-1)|,\\
    |S_2(\ell)|&=|S_2(\ell-1)|.
\end{align*}
\end{lemma}
\versionShortLong{}{
\begin{proof}
  To obtain valid $S_t(\ell-1)$, we require $r<q^{\ell-1}=\frac{q}{2}$.
 According to Lemma~\ref{lem:go-down} and \cref{lem:down-by-one}, we know that if $(a,b)\in S_0(\ell)$, then $(a',b')\in S_0(\ell-1)\cup S_1(\ell-1)\cup S_2(\ell-1)$. The statement can be proven by counting how many ways to add the most significant bits $a_\ell$ and $b_\ell$ for $a'$ and $b'$ to obtain $a$ and $b$. Denote them by $a=[a',a_\ell],b=[b',b_\ell]$. Recall the definition in~\eqref{eq:def_S_t}, given $(a',b')\in S_t(\ell-1)$, there exist $i'\leq_2 b',j'\leq_2 b'-i'$ such that $2i'+j'+a'=\frac{q}{2}-r'+t\frac{q}{2}$ with $r'\in[r]$.
 Construct $i,j$ by appending one most significant bit to $i',j'$, i.e., $i=[i',i_\ell]$ and $j=[j',j_\ell]$ with $i_\ell\leq_2b_\ell$ and $j_\ell\leq_2b_\ell$. To obtain $(a,b)\in S_t(\ell)$, we require $2i+j+a=q-r'' + tq$ with $r''\in[r]$.
We can write
\begin{align*}
  2i+j+a=2i'+j'+a' + (2i_\ell+j_\ell+a_\ell)\frac{q}{2}\ .
\end{align*}
Since the difference between $2i+j+a$ and $2i'+j'+a'$ is always some multiple of $\frac{q}{2}$, $r''=r'$.

Recall that $S_2(\ell-1)\subset S_1(\ell-1)\subset S_0(\ell-1)$ from Lemma~\ref{lem:down-by-one},

We first prove $|S_0(\ell)|$. To have $(a,b)\in S_0(\ell)$, we require $2i+j+a=q-r'$. Consider three cases,
\begin{itemize}
\item
  Given $(a',b')\in S_0(\ell-1) \setminus S_1(\ell-1)$, it means 
  $2i'+j'+a'=\frac{q}{2}-r'$. To obtain $2i+j+a=q-r'$, we require $2i_\ell+j_\ell+a_\ell=1$. There are three options of $(a_\ell,b_\ell)$ that this can be fulfilled, i.e., $(a_\ell,b_\ell)=(1,0), (0,1)$ or $(1,1)$.
\item
  Given $(a',b')\in S_1(\ell-1)\setminus S_2(\ell-1)$, we have $2i'+j'+a'=q-r'$, the option $(0,0)$ for the most significant bit $(a_\ell,b_\ell)$ allow to get $(a,b)\in S_0(\ell)$. Since $S_1(\ell-1)\subset S_0(\ell-1)$, we can find $i''\leq_2 i'\leq_2 b'$ and $j''\leq_2 j'\leq_2 b-i$ such that $2i''+j''+a'=\frac{q}{2}-r'$ (e.g., by~\cref{algo:ij-reduction}). So all the other three options in the first case are also valid for this case.
\item
  Given $(a',b')\in S_2(\ell-1)$, we have $2i'+j'+a'=\frac{3}{2}q-r'$. Since $S_2(\ell-1)\subset S_1(\ell-1)$, all four options of $(a_\ell,b_\ell)$ allow to get $(a,b)\in S_0(\ell)$.
\end{itemize}

Then we show $|S_1(\ell)|$. We require $2i+j+a=2q-r'$. Again, consider the three cases,
\begin{itemize}
\item
  Given $(a',b')\in S_0(\ell-1) \setminus S_1(\ell-1)$, we have $2i'+j'+a'=\frac{q}{2}-r'$. This means that $2i_\ell+j_\ell+a_\ell=3$ is required. $(a_\ell,b_\ell)=(1,1)$ is the only way to add the most significant bits. 
\item
  Given $(a',b')\in S_1(\ell-1)\setminus S_2(\ell-1)$, we have $2i'+j'+a'=q-r'$. We require $2i_\ell+j_\ell+a_\ell=2$ to obtain $2i+j+a=2q-r'$. The two options $(0,1)$ and $(1,1)$ allow this.
\item
  Given $(a',b')\in S_2(\ell-1)$, we have $2i'+j'+a'=\frac{3}{2}q-r'$. We require $2i_\ell+j_\ell+a_\ell=1$ to obtain $2i+j+a=2q-r'$. The three options $(a_\ell,b_\ell)=(1,0), (0,1)$ or $(1,1)$ can fulfill this.
\end{itemize}
Now we show $|S_2(\ell)|$. We require $2i+j+a=3q-r'$. Consider the three cases,
\begin{itemize}
\item
  Given $(a',b')\in S_0(\ell-1) \setminus S_1(\ell-1)$, we have $2i'+j'+a'=\frac{q}{2}-r'$. This means that $2i_\ell+j_\ell+a_\ell=5$ is required. However, due to $j_\ell\leq_2 b_\ell -i_\ell$, this cannot happen since $i_\ell$ and $j_\ell$ cannot be one at the same time.
\item
  Given $(a',b')\in S_1(\ell-1)\setminus S_2(\ell-1)$, we have $2i'+j'+a'=q-r'$. We require $2i_\ell+j_\ell+a_\ell=4$. However this cannot happen since $i_\ell$ and $j_\ell$ cannot be one at the same time.
\item
  Given $(a',b')\in S_2(\ell-1)$, we have $2i'+j'+a'=\frac{3}{2}q-r'$. We require $2i_\ell+j_\ell+a_\ell=3$. $(a_\ell,b_\ell)=(1,1)$ is the only option.
\end{itemize}
To sum up, the statements follow from
\begin{align*}
  |S_0(\ell)|&=3(|S_0(\ell-1)\setminus S_1(\ell-1)|)+4(|S_1(\ell-1)\setminus S_2(\ell-1)|)+ 4|S_2(\ell-1)|\\
  |S_1(\ell)|&=|S_0(\ell-1)\setminus S_1(\ell-1)|+2|S_1(\ell-1)\setminus S_2(\ell-1)| +3|S_2(\ell-1)|\\
  |S_2(\ell)|&=|S_2(\ell-1)|\ .
\end{align*}
\qed \end{proof}}
Lemma~\ref{lem: key ingredient} yields a recurrence relation for $|S_0(\ell)|$, $|S_1(\ell)|$ and $|S_2(\ell)|$. For a given $r$, the initial value $\ell_0$ should be chosen such that $S_i(\ell_0), i=0,1,2$ is a valid set according to the definition in~\eqref{eq:def_S_t}.
\new{Denote by $\bs(\ell) = (|S_0(\ell)|,|S_1(\ell)|,|S_2(\ell)|)^\top$.}
We then have
\begin{align}\label{eq:matA}
\new{\bs(\ell)= \bA^{\ell-\ell_0}\cdot \bs(\ell_0),\ \text{ where } \bA = \begin{pmatrix}
3 & 1 & 0\\
1 & 1 & 1\\
0 & 0 & 1
\end{pmatrix}.
}
\end{align}
The \new{recursion} enables us to find the asymptotic behavior of the number of $(\Phi,q-r)$-bad monomials, which is exactly $|S_0(\ell)|$.
Note that the order of $|S_j(\ell)|, j=0,1,2$ is controlled by $\lambda_1^\ell$, where $\lambda_1=2+\sqrt{2}$ is the largest eigenvalue of $\bA$ \new{in~\eqref{eq:matA}}.
Hence,
\begin{align}\label{eq:S0_asym}
  |S_0(\ell)|=\Theta ((2+\sqrt{2})^{\ell}).
\end{align}
For different $r$, the exact values of $|S_0(\ell)|$ can be different, since the initial value $|S_0(\ell_0)|$ depends on $r$. However, the asymptotic behavior is the same for any fixed $r$.

We provide the exact expressions of $|S_0(\ell)|$ for $r=1$ and $r=3$, denoted by $|S_0^{(1)}(\ell)|$ and $|S_0^{(3)}(\ell)|$ respectively, which we will later use to derive upper and lower bound on the number of $(\Phi,q-r)^*$-bad monomials:
\begin{align}
    |S_0^{(1)}(\ell)|= & \frac{5\sqrt{2}+7}{2(3\sqrt{2}+4)} \cdot \lambda_1^\ell + \frac{5\sqrt{2}-7}{2(3\sqrt{2}-4)} \cdot \lambda_2^\ell \nonumber \\
    \approx & 0.8536 \cdot \lambda_1^\ell +  0.1464 \cdot \lambda_2^\ell \label{eq:S0r1}\\
    |S_0^{(3)}(\ell)|= & \frac{65\sqrt{2}+92}{4(12\sqrt{2}+17)}\cdot \lambda_1^\ell + \frac{65\sqrt{2}-92}{4(12\sqrt{2}-17)} \cdot \lambda_2^\ell - \lambda_3^\ell \nonumber \\
    \approx & 1.3536 \cdot \lambda_1^\ell +  0.6465\cdot \lambda_2^\ell -1 \label{eq:S0r3}
\end{align}
where $\lambda_1=2+\sqrt{2},\lambda_2=2-\sqrt{2},\lambda_3=1$ are \new{the} three distinct eigenvalues of the matrix $\bA$.

\subsubsection{Counting $(\Phi,q-r)^*$-bad monomials: }
For $q=2^\ell$ and $1\leq r<q$, we define the following set 
\begin{align*}
  S^*(\ell):=\bigg\{(a,b)\in \mathbb{Z}_q^2\ :\ \begin{split}
        \exists\ i\leq_2 b, j\leq_2 b-i ,
        \textrm{ s.t.}&\ 2i+j+a=q-r'+t(q-1),\\& \textrm{ for some }r'\in[r], t\geq 0
    \end{split}\ \bigg\}\ .
\end{align*}
It is clear that $(a,b)\in S^*(\ell)$ if and only if $x^ay^b$ is $(\Phi,q-r)^*$-bad.

We first relate the value $|S^*(\ell)|$ to $|S_0(\ell)|$ in Lemma~\ref{lem:upper_bad} and Lemma~\ref{lem:lower_bad}.
\begin{lemma}\label{lem:upper_bad}
    Let $\ell\geq 2, q=2^\ell$, $1\leq r\leq \frac{q}{4}$, $s=\ceil{\log_2 (r)}$ and $q'=2^{\ell-s}$. Denote by $S^{(3)}_0(\ell-s)$  the set of $(a,b)$ such that $x^ay^b$ is $(\Phi,q'-3)$-bad. Then
    \begin{align}
      |S^*(\ell)|< 4r^2\cdot |S_0^{(3)}(\ell-s)|. \nonumber
    \end{align}
    If $r$ is a power of $2$, then
    \begin{align}
    |S^*(\ell)|\leq r^2\cdot |S_0^{(3)}(\ell-s)|. \nonumber
  \end{align}
\end{lemma}
\versionShortLong{}{
\begin{proof}
By definition, we require $\ell-s\geq 2$ to have a valid $S_0^{(3)}(\ell-s)$. Therefore, we require $\ell\geq 2$ and $r\leq \frac{q}{4}$.
Let $x^a y^b$ be an arbitrary $(\Phi,q-r)^*$-bad monomial.
By definition, this means that there exist $i\le_2 b$ and $j\le_2 b-i$ such that $2i+j+a = q-r' + (q-1)t$ for some $1\leq r'\leq r$ and $0\leq t \leq 2$.\footnote{Note that $2i+j+a\leq 2b+a\leq 3(q-1)<q-r+(q-1)t$ for any $t\geq 3$ and $r<q$.} We drop $s=\lceil \log (r)\rceil$ least significant bits in $a$, $b$, $i$ and $j$ to obtain $a'$, $b'$, $i'$ and $j'$,
i.e., one can write
\begin{align*}
    i&=i'\cdot 2^s+r_{i},\\
    2i&=2i'\cdot 2^s + 2r_i,\\
    j&=j'\cdot 2^s + r_j,\\
    a&=a'\cdot 2^s + r_a,
\end{align*}
where the remainders $0\leq r_{i},r_j,r_a<2^s$.
Denote by $q'=q/{2^s}=2^{\ell-s}$, it is clear that
\begin{align}
    2i'+j'+a'&= \frac{2i+j+a}{2^s}-\frac{2r_{i}+r_j+r_a}{2^s}\nonumber\\
    &=\frac{q-r'+(q-1)t}{2^s}-\frac{2r_{i}+r_j+r_a}{2^s}\nonumber \\
    &=q'(t+1)-\frac{r'+t}{2^s}-\frac{2r_{i}+r_j+r_a}{2^s}\nonumber
\end{align}
Since the bits in $i,j$ cannot be both one at the same position, $2r_i+r_j\leq 2(2^s-1)$. Hence $0\leq 2r_i+r_j+r_a\leq 3(2^s-1)$. In addition, since $1\leq r'+t\leq r+2\leq 2^s+2$, we then have
$$
q'(t+1)-4< 2i'+j'+a'\leq q'(t+1) - \frac{1}{2^s}.
$$
As $2i'+j'+a'$ can only be integer, we have
\begin{align*}
    q'(t+1)-3\leq 2i'+j'+a'\leq q'(t+1) - 1.
\end{align*}
This implies that $(a',b')$ is $(\Phi,q'-3)$-bad since $i'\le_2 b'$ and $j'\le_2 b'-i'$. Therefore, adding arbitrary $s$ least significant bits to a pair $(a',b')\in  S_0^{(3)}(\ell-s)$, the obtained $(a,b)$ may be $(\Phi,q-r)^*$-bad. The number of $(\Phi,q-r)^*$-bad monomials can therefore be bounded from above by
\begin{align*}
2^s\cdot  2^s\cdot |S_0^{(3)}(\ell-s)|&= (2^{\ceil{\log_2(r)}})^2  \cdot |S_0^{(3)}(\ell-s)|\\
&< (2r)^2 \cdot |S_0^{(3)}(\ell-s)|\ .
\end{align*}

If $r$ is a power of $2$, we can set $s=\log_2r$ and obtain the tighter bound.
\qed \end{proof}}
\begin{lemma}\label{lem:lower_bad}
  Let $\ell\geq 1, q=2^\ell$, $1\leq r\leq \frac{q}{2}$, $s=\floor{\log_2 r}$ and $q'=2^{\ell-s}$. Denote by $S_0^{(1)}(\ell-s)$ the set of $(a,b)$ such that $x^ay^b$ is $(\Phi,q'-1)$-bad. Then
  \begin{align}
    |S^*(\ell)|> \frac{r^2}{4}\cdot |S_0^{(1)}(\ell-s)|\ . \nonumber
  \end{align}
  If $r$ is a power of $2$, then
  \begin{align}
    |S^*(\ell)|\geq  r^2\cdot |S_0^{(1)}(\ell-s)|\ . \nonumber
  \end{align}
\end{lemma}
\versionShortLong{}{
\begin{proof}
 It follows from Lemma~\ref{lem:down-by-one} that $x^ay^b$ is $(\Phi,q'-1)$-bad if and only if $(a,b)\in S_0^{(1)}(\ell-s)$. By definition, we require $\ell-s\geq 1$ to have a valid $S_0^{(1)}(\ell-s)$. Therefore, we require $\ell\geq 1$ and $r\leq \frac{q}{2}$. Consider a pair $(a',b')\in S_0^{(1)}(\ell-s)$. According to the definition~\eqref{eq:def_S_t}, there exist $i'\leq_2 b,j'\leq_2 b-i$ such that $2i'+j'+a'=q'-1$. Construct integers $a,b,i,j$ by appending $s$ least significant bits to the binary representation $\ba', \bb',\bi',\bj'$ of $a',b',i',j'$ respectively, i.e., $a:=(a'',\ba')_2,b:=(b'',\bb')_2,i:=(i''=0,\bi')_2,j:=(j''=0,\bj')_2$ with $a'',b''\in\{0,\ldots,2^s-1\}$. It can be seen that
\begin{align*}
    2i+j+a&=(2i'+j'+a')2^s + 2i''+j''+a''\\
    &=(q'-1)\cdot 2^s+a''\\
    &= q- 2^s + a''.
\end{align*}
We can choose $a''$  and $b''$ to be any integer of $s$-bits  such that $q-r\leq 2i+j+a\leq q-1$. Namely, given a pair $(a',b')\in S_0(\ell-s)$ with $s=\floor{\log_2 r}$, we have in total $(2^s)^2 > \left(\frac{r}{2}\right)^2$ ways of choosing $a'',b''$ such that $x^ay^b$ is $(\Phi,q-r)^*$-bad. 
If $r$ is a power of $2$, we can set $s=\log_2r$ and obtain a tighter lower bound.
\qed \end{proof}}
In the following theorem we provide the exact expressions of upper and lower bounds on $|S^*(\ell)|$), using the exact expression of $|S_0(\ell)|$ in \eqref{eq:S0r1} and \eqref{eq:S0r3}.
\begin{theorem}\label{thm:exact_bounds_bad}
    Let $\ell\geq 2, q=2^\ell, 1\leq r\leq \frac{q}{4}$ and $s = \log_2r$, the number of $(\Phi,q-r)^*$-bad monomials is
    \begin{align}
        \frac{0.8536 \cdot \lambda_1^{\ell-\floor{s}} +  0.1464 \cdot \lambda_2^{\ell-\floor{s}}}{4}< \frac{|S^*(\ell)|}{r^2}< 4(1.3536 \cdot \lambda_1^{\ell-\ceil{s}} +  0.6465\cdot \lambda_2^{\ell-\ceil{s}} -1)\ , \nonumber
    \end{align}
    where $\lambda_1 = 2+\sqrt{2}$ and $\lambda_2 = 2-\sqrt{2}$.

    If $r$ is a power of $2$, we obtain
    \begin{align}
        0.8536 \cdot \lambda_1^{\ell-s} +  0.1464 \cdot \lambda_2^{\ell-s}
        \leq \frac{|S^*(\ell)|}{r^2}\leq 1.3536 \cdot \lambda_1^{\ell-s} +  0.6465\cdot \lambda_2^{\ell-s} -1\ .
        \nonumber
    \end{align}
\end{theorem}
\versionShortLong{}{
\begin{proof}
It follows directly from the estimation of $|S_0(\ell)|$ in \eqref{eq:S0r1} -- \eqref{eq:S0r3} and the bounds in \cref{lem:upper_bad} and \cref{lem:lower_bad}.
\end{proof}}

We can then derive an asymptotic behavior of the rate of QC-LRS codes in~\cref{thm:asym_rate}.
\begin{corollary}\label{thm:asym_rate}
  Let $\mu = \log_2(2+\sqrt{2})$. For $q\to\infty$ and $1\leq r\leq \frac{q}{4}$, the number of $(\Phi,q-r)^*$-bad monomials is 
  \begin{align}
  |S^*(\ell)|=\Theta(r^{2-\mu} q^{\mu})\ .\nonumber
  \end{align}
  Further, the QC-LRS code $\cC_q(\Phi,q-r)$ has rate
  \begin{align*}
    R= 1-\Theta\left((q/r)^{\mu-2}\right)=1-\Theta\left((q/r)^{-0.2284}\right).
  \end{align*}
\end{corollary}
\versionShortLong{}{
\begin{proof}
It can be seen from \cref{thm:exact_bounds_bad} that the order of $|S^*(\ell)|$ is controlled by $\lambda_1^\ell$. The asymptotic estimation is obtained by neglecting the other terms and the constant coefficients.
The rate is calculated by the number of good monomials, which is $q^2-|S^*(\ell)|$, dividing the number of all bi-variate monomials, which is $q^2$.
\qed \end{proof}}
For an illustration, we plot in~\cref{fig:dimPlot1} the dimension of the code $\cC_q(\Phi,q-r)$ with $q=2^5$, which is done by computer-search according to the necessary and sufficient condition in~\cref{lem:goodCond}, and the corresponding lower and upper bounds for $r\in[1,q/4]$ based on the bounds on $|S^*(\ell)|$ in \cref{thm:exact_bounds_bad}.
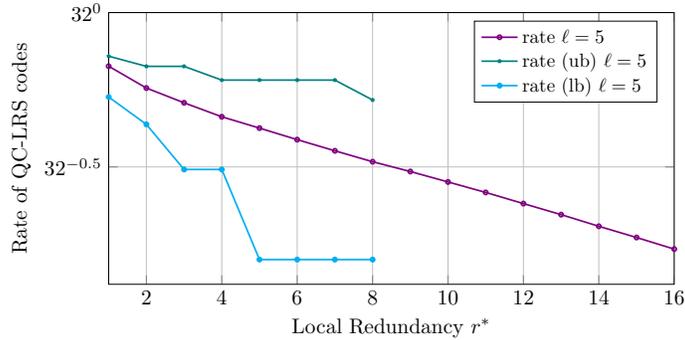
\begin{figure}[h]
  \centering
  \input{ratePlot_l=5_log.tex}
  \caption{The dimension of QC-LRS code $\cC_q(\Phi,q-r)$ with $q=2^5$ along with the corresponding upper bound (ub) and lower bound (lb) for $r\in[1,q/4]$ calculated by $1-|S^*(\ell)|/q^2$. The lower and upper bound on $|S^*(\ell)|$ are given in \cref{thm:exact_bounds_bad}.
\vspace{-5ex}  }
  \label{fig:dimPlot1}
\end{figure}
\begin{remark}\label{rem:LRSasymRate}
  Recall that the rate of bivariate lifted Reed-Solomon (LRS) codes is $R= 1- \Theta((q/r)^{\log_2 3 -2}=1-\Theta((q/r)^{-0.4150})$~\cite{holzbaur2020lifted}. We compare the performance of our codes with LRS codes in terms of local recovery in an erasure channel in~\cref{sec:local-correction}.
\end{remark}
\subsection{Distance of Quadratic-Curve-Lifted RS Codes}
We provide the upper and lower bounds on the distance of the QC-LRS codes $\cC_{q}(\Phi,q-r)$ in the following theorem.
\begin{theorem}[Bounds on the Minimum Distance]
  Let $q$ be a power of $2$ and $\Phi$ be the set of all quadratic functions. The QC-LRS code $\cC_{q}(\Phi,q-r)$ has minimum distance
  \begin{align*}
    qr+1\leq\dist(\cC_{q}(\Phi,q-r))\leq qr+q\ .
  \end{align*}
\end{theorem}
\begin{proof}
  \versionShortLong{
  The upper bound is proven by counting the number of zeros in the codeword $f(x,y)=\prod_{\alpha\in\cA} (x-\alpha)$, where $\cA$ is a subset of $\Fq$ with $|\cA|=q-r-1$. The lower bound is proven by considering the minimum number of non-zero positions in all disjoint local groups (e.g., all lines) of a non-zero symbol in a codeword. For a detailed proof we refer to the full version of this paper~\cite{liu2021quadratic}.}{
We first show the upper bound. Let $\cA\subset \Fq$ be a subset with $|\cA|=q-r-1$. Consider a function $f(x,y)=\prod_{\alpha\in\cA} (x-\alpha)$. It can be seen that $\deg(f|_\phi(x))=q-r-1$ for any $\phi\in\Phi$ therefore $f(x,y)$ is in the code $\cC_{q}(\Phi,q-r)$. The zeros of $f(x,y)$ in $\Fq^2$ are $\{(x,y):x\in\cA,y\in\Fq\}$. Therefore, the evaluations of $f(x,y)$ in $\Fq^2$ is of weight $q^2-q(q-r-1)$. Due to the linearity of the code, the upper bound on the minimum distance is proven.\\
Now we prove the lower bound. For any non-zero $f\in\cC_{q}(\Phi,q-r)$ consider a point $\bp\in\Fq^2$ such that $f(\bp)\neq 0$. Denote by $\cL_{\bp,1}\subset\Phi$ the set of lines in $\Phi$ intersecting with each other only at $\bp$. It can be seen that $|\cL_{\bp,1}|=q$. By definition, $\deg f|_\phi<q-r$ for any $L\in\cL_{\bp,1}$, therefore there are at most $q-r-1$ zeros in the evaluations of $f$ on $\phi$.
Denote by $\wt(f)$ the number of non-zero evaluations of $f$ on $\Fq^2$ and by $\wt(f|_L)$ the number of non-zero evaluations of $f$ on $\phi$, then
\begin{align*}
    \wt(f)&\geq \sum_{L\in\cL_{\bp,1}}(\wt(f|_L)\underbrace{-1}_{\textrm{excluding }f(\bp)}) \underbrace{+ 1}_{\textrm{including }f(\bp)}\\
    &\geq qr+1
\end{align*}}
\new{Note that the bounds are derived in a similar method as for LRS codes in \cite[Theorem 5.1]{guo2013new}.}
\qed \end{proof}
\versionShortLong{}{
\begin{remark}
Note that the vertical lines $x=\eta$ (constant) are not included in the set of quadratic curves $\Phi$, which is the reason why the lower bound is worse than the lower bound $d\geq (q+1)r+1$ for lifted Reed-Solomon codes~\cite{guo2013new}.
\end{remark}}

\section{Local Recovery Capability from Erasures}
\label{sec:local-correction}
\new{For a code with locality~\cite{tamo2014family}, the local groups of a codeword symbol are defined as the sets of indices where the symbol can be recovered by accessing only the symbols in one of the sets.}
Given a QC-LRS code over $\Fq$, the number of \new{local recovery sets} of any codeword symbol is the number of quadratic curves over $\Fq$ \new{passing through a certain point}, which is $q^2$. For an LRS codes, the number of \new{local recovery sets} is $q+1$.
Consider an erasure channel with erasure probability $\tau$. With respect to the local recovery, we are interested in correcting a certain erasure within a \new{local recovery set} and how large the failure probabilities of LRS/QC-LRS codes is. The failure probability is exactly the probability that there are at least $r$ other erasures in each \new{local recovery set} of the erased symbol to be recovered. For LRS codes, since all the \new{local recovery sets} are disjoint, the failure probability is exactly $\left(\sum_{i=r}^{q-1} \binom{q-1}{i} \tau^i(1-\tau)^{q-1-i} \right)^{q+1}$.
For QC-LRS codes, since the \new{local recovery sets} may intersect with each other, an analysis for the closed form of the failure probability is still an open problem.
In order to compare the performance of these two codes, we run simulations with both codes of length $n=64$, dimension $k=10$ and $k=6$, respectively. The simulation results are presented in \cref{fig:dim10}. We can see that for both $\dim=10$ and $\dim=6$, the failure probability of local recovery with QC-LRS is smaller than or similar to that with LRS codes for $\tau\leq 0.7$.
Therefore, for this length, QC-LRS codes perform better than LRS for local recovery.
\begin{figure}[h]
  \centering
    \input{global_l=3_dim10.tex}
  \caption{Local recovery performance of LRS/RM and QC-LRS \new{($\cC_{q}(\Phi,d=q-r)$)} codes of length $n=q^2=64$ dimension $k=10$ (rate $=k/n=0.15625$) and dimension $k=6$ (rate $=k/n=0.09375$). Note that the LRS codes with the parameters here are RM codes. \vspace{-2ex}}
  \label{fig:dim10}
\end{figure}
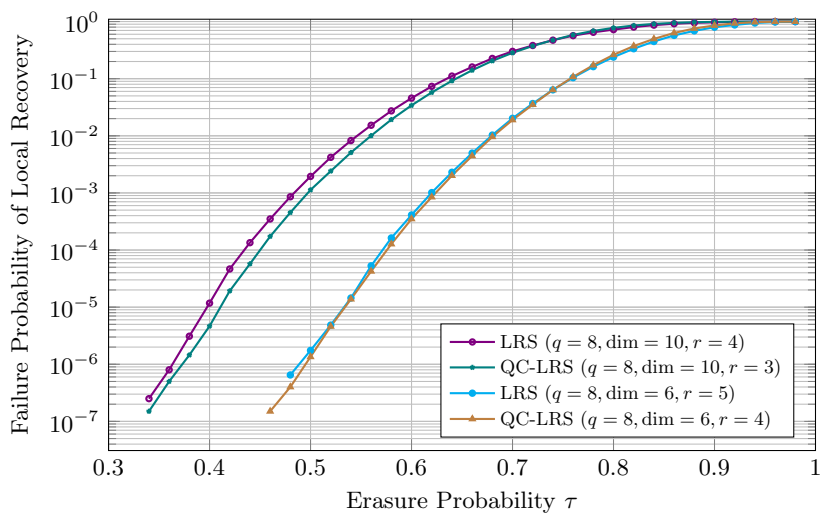
\vspace{-1em}

\bibliographystyle{splncs04}
\bibliography{refs}


\end{document}

%% file: defs.tex
\usepackage{indentfirst}
\usepackage[utf8]{inputenc}
\usepackage{amssymb,amsfonts}
\usepackage{amsmath}
\usepackage{xcolor}
\usepackage{graphicx}
\usepackage{caption}
\usepackage{subcaption}
\usepackage[linesnumbered,ruled]{algorithm2e}
\usepackage[normalem]{ulem}

\usepackage[export]{adjustbox}
\usepackage[absolute, overlay]{textpos} 

\usepackage{tikz}
\usepackage{pgfplots}
\pgfplotsset{compat=newest}
\usetikzlibrary{arrows, shapes, shapes.multipart, trees, positioning, calc,
  backgrounds, fit, matrix, decorations.pathmorphing,patterns}

\pgfplotscreateplotcyclelist{sims_list}{%
{violet,solid, thick, mark=o, mark size=1pt},
{teal,solid, thick, mark=star,mark size=1pt},
{cyan,solid, thick, mark=otimes,,mark size=1pt},
{brown,solid, thick, mark=triangle,mark size=1pt},
{olive,solid, thick, mark=o,mark size=1pt},
{orange, solid, thick, mark=diamond,mark size=1pt},
{lime,solid, thick, mark=square,mark size=1pt},
}

\usepackage{hyperref}
\usepackage[capitalise]{cleveref}

\newif\ifhighlightShortLongDiffs
\newif\ifversionShort
\newcommand{\versionShortLong}[2]{\ifversionShort \ifhighlightShortLongDiffs \textcolor{green}{Short version:} \fi #1 \else \ifhighlightShortLongDiffs \textcolor{green}{Long version:}\fi #2\fi}
\newif\ifcomment

\newcommand{\new}[1]{\ifcomment \textcolor{blue}{#1}\fi #1}

\DeclareMathOperator{\wt}{wt}

\DeclareMathOperator{\dist}{dist}

\newcommand{\F}{\ensuremath{\mathbb{F}}}
\newcommand{\Fq}{\ensuremath{\mathbb{F}_q}}

\newcommand{\ceil}[1]{\ensuremath{\left\lceil{#1}\right\rceil}}
\newcommand{\floor}[1]{\ensuremath{\left\lfloor{#1}\right\rfloor}}
\renewcommand{\leq}{\leqslant}
\renewcommand{\geq}{\geqslant}

\newcommand{\Mod}{\mathrm{mod}}
\newcommand{\modstarq}{\ (\Mod^*\ q)}

\newcommand{\bbN}{\ensuremath{\mathbb{N}}}

\newcommand{\cA}{\ensuremath{\mathcal{A}}}
\newcommand{\cB}{\ensuremath{\mathcal{B}}}
\newcommand{\cC}{\ensuremath{\mathcal{C}}}

\newcommand{\cG}{\ensuremath{\mathcal{G}}}

\newcommand{\cL}{\ensuremath{\mathcal{L}}}

\renewcommand{\vec}[1]{\ensuremath{\boldsymbol{#1}}}

\newcommand{\bA}{\vec{A}}
\newcommand{\ba}{\vec{a}}

\newcommand{\bb}{\vec{b}}

\newcommand{\bi}{\vec{i}}

\newcommand{\bj}{\vec{j}}

\newcommand{\bp}{\vec{p}}

\newcommand{\bs}{\vec{s}}

\usepackage{xcolor}
\definecolor{brightmaroon}{rgb}{0.76, 0.13, 0.28}
\definecolor{ao}{rgb}{0.0, 0.5, 0.0}
\definecolor{azure}{rgb}{0.0, 0.5, 1.0}
\definecolor{TUMBlue}{RGB}{0,101,189} 
\definecolor{TUMBlueDark}{RGB}{0,82,147} 
\definecolor{TUMBlueDark2}{RGB}{0,51,89} 
\definecolor{TUMBlueLight}{RGB}{152,198,234} 
\definecolor{TUMBlueLight2}{RGB}{000,115,207} 
\definecolor{TUMBlueMiddle}{RGB}{100,160,200} 

\definecolor{TUMElfenbein}{RGB}{218,215,203} 
\definecolor{TUMGreen}{RGB}{162,173,0} 
\definecolor{TUMGreenLight}{RGB}{0,124,48}
\definecolor{TUMGreenLight2}{RGB}{103,154,29}
\definecolor{TUMOrange}{RGB}{227,114,34} 
\definecolor{TUMOrangeLight}{RGB}{239,136,34} 
\definecolor{TUMOrangerDark}{RGB}{214,076,019}

\definecolor{TUMGrayDark}{RGB}{088,088,090}
\definecolor{TUMGray}{RGB}{156,157,159}
\definecolor{TUMGrayLight}{RGB}{215,217,218}

\definecolor{TUMRed}{RGB}{196,7,27}
\definecolor{TUMRedDark}{RGB}{156,013,022}
\definecolor{TUMRedLight}{RGB}{236,37,027}

\definecolor{TUMpurple}{RGB}{105,008,090}
\definecolor{TUMviolet}{RGB}{015,027,095}
\definecolor{TUMcyan}{RGB}{000,119,138}	
\definecolor{TUMyellow}{RGB}{255,220,000}
\definecolor{TUMgold}{RGB}{249,186,000}
\definecolor{TUMgoldhigh}{RGB}{255,200,0}


%% file: ratePlot_l=5_log.tex
    \begin{tikzpicture}[font=\normalsize,scale = 0.8]
    \pgfplotsset{compat = 1.3}
    \begin{axis}[
    legend style={nodes={scale=0.9, transform shape}},
    cycle list name = {sims_list},
    width = 0.9\columnwidth,
    height = 0.5\columnwidth,
    xlabel = {{Local Redundancy $r^*$}},
    ylabel = {{Rate of QC-LRS codes}},
    ymode=log,
    log basis y={32},
    xmin = 1,
    xmax = 16,
    ymin = 0,
    ymax = 1,
    legend pos = north east,
    legend cell align=left,
    grid=both]

\addplot table[x=r, y=rate] {dim_l=5_CL.dat};

\addlegendentry{{rate $\ell=5$}};



\addplot table[x=r, y=rate_ub] {dimBounds_l=5_CL.dat};

\addlegendentry{{rate (ub) $\ell=5$}};

\addplot table[x=r, y=rate_lb] {dimBounds_l=5_CL.dat};

\addlegendentry{{rate (lb) $\ell=5$}};
    \end{axis}
    \end{tikzpicture}

%% file: global_l=3_dim10.tex
    \begin{tikzpicture}
    \pgfplotsset{compat = 1.3}
    \begin{semilogyaxis}[
    legend style={nodes={scale=0.8, transform shape}},
    cycle list name = {sims_list},
    width = 0.9\columnwidth,
    height = 0.6\columnwidth,
    xlabel = {{Erasure Probability $\tau$}},
    ylabel = {{Failure Probability of Local Recovery}},
    xmin = 0.3,
    xmax = 1,
    ymin = 0,
    ymax = 1.1,
    legend pos = south east,
    legend cell align=left,
    grid=both]
\addplot table[x=tau,y=fail_rate] {dim10_fail_rate_l=3_r=4_step=50_LRS.dat};

\addlegendentry{{LRS $(q=8, \dim=10, r=4)$}};

\addplot table[x=tau,y=fail_rate] {dim10_fail_rate_l=3_r=3_step=50_CL.dat};

\addlegendentry{{QC-LRS $(q=8, \dim=10, r=3)$}};

\addplot table[x=tau,y=fail_rate] {dim6_fail_rate_l=3_r=5_step=50_LRS.dat};

\addlegendentry{{LRS $(q=8, \dim=6, r=5)$}};

\addplot table[x=tau,y=fail_rate] {dim6_fail_rate_l=3_r=4_step=50_CL.dat};

\addlegendentry{{QC-LRS $(q=8, \dim=6, r=4)$}};




    \end{semilogyaxis}
    \end{tikzpicture}